\documentclass[superscriptaddress]{revtex4}
\usepackage{amsmath,amsthm,amsfonts}
\newtheorem{Thm}{Theorem}
\theoremstyle{definition}
\newtheorem{Ex}{Example}
\usepackage{graphicx}
\DeclareMathOperator{\supp}{supp}
\begin{document}
\title{$q$-Generalization of the inverse Fourier transform}

\author{M. Jauregui}
\affiliation{Centro Brasileiro de Pesquisas Fisicas and National Institute of Science and Technology for Complex Systems, Rua Xavier Sigaud 150, 22290-180 Rio de Janeiro, Brazil}

\author{C. Tsallis}
\affiliation{Centro Brasileiro de Pesquisas Fisicas and National Institute of Science and Technology for Complex Systems, Rua Xavier Sigaud 150, 22290-180 Rio de Janeiro, Brazil}
\affiliation{Santa Fe Institute, 1399 Hyde Park Road, Santa Fe, New Mexico 87501, USA}
\begin{abstract}
A wide class of physical distributions appears to follow the $q$-Gaussian form, which plays the role of attractor according to a $q$-generalized Central Limit Theorem, where a $q$-generalized Fourier transform plays an important role. We introduce here a method which determines a distribution from the knowledge of its $q$-Fourier transform and some supplementary information. This procedure involves a recently $q$-generalized representation of the Dirac delta and the class of functions on which it acts. The present method conveniently extends the inverse of the standard Fourier transform, and is therefore expected to be very useful in the study of many complex systems.
\end{abstract}

\maketitle
\section{Introduction}
Nonextensive statistical mechanics \cite{Tsallis1988,GellMannTsallis2004}, a current generalization of the Boltzmann-Gibbs (BG) theory, is actively studied in diverse areas of physics and other sciences. This theory is based on a nonadditive (though extensive \cite{TsallisGellMannSato2005}) entropy characterized by the index $q$, such that $q=1$ recovers the standard BG entropy. It has been applied in systems such as cold atoms in dissipative optical lattices \cite{DouglasBergaminiRenzoni2006}, dusty plasmas \cite{LiuGoree2008}, trapped ions \cite{DeVoe2009}, spin-glasses \cite{PickupCywinskiPappasFaragoFouquet2009}, turbulence in the heliosheath \cite{BurlagaNess2009}, self-organized criticality \cite{CarusoPluchinoLatoraVinciguerraRapisarda2007}, high-energy experiments at LHC/CMS/CERN \cite{CMS1} and RHIC/PHENIX/Brookhaven \cite{PHENIX}, low-dimensional dissipative maps \cite{LyraTsallis1998}, finance \cite{Borland2002}, among others.  
 
This new theory can be advantageously based on $q$-gen\-e\-ral\-i\-za\-tions of standard mathematical concepts, such as the logarithm and exponential functions, addition and multiplication, Fourier transform (FT) and the Central Limit Theorem (CLT) \cite{UmarovTsallisSteinberg2008}. Recently, plane waves, and the representation of the Dirac delta in plane waves have been generalized as well \cite{JaureguiTsallis2010,ChevreuilPlastinoVignat2010,Mamode,PlastinoRocca}. 
Some of these generalizations open the door to interesting aspects. 
For instance, a generic analytical expression for the inverse $q$-FT for arbitrary functions and any value of $q$ does not exist \cite{Hilhorst2010}. The focus on this fact and related questions should be relevant for various applications in physics (e.g., field theory and condensed matter physics), engineering (e.g., image and signal processing), and mathematics where the standard FT and its inverse play a crucial role. 

We show in this paper that, in the $1\le q<2$ particular case and for non-negative functions (e.g., probability distributions), it is possible, by using special information that we shall detail later on, to obtain a biunivocal relation between the function and its $q$-FT. 
This property, not to be confused with an inverse $q$-FT,  is not relevant to the proof in \cite{UmarovTsallisSteinberg2008} of the $q$-generalized CLT, which yields $q$-Gaussian attractors (defined here below). 
It should be also clear that such an inverse $q$-FT is by no means necessary for the existence of attractors, which can be proved \cite{HahnJiangUmarov2010} without recourse to this (generically nonlinear) integral transform.

The $q$-FT of a non-negative integrable function $f(x)$ of the real variable $x$ is defined as \cite{UmarovTsallisSteinberg2008}
\begin{equation}
F_q[f](\xi)=\int_{-\infty}^{\infty} f(x)e_q^{i \xi x[f(x)]^{q-1}}\, dx\qquad(1\le q<3)\,,
\label{qFT}
\end{equation}
where $e_q^z$ denotes the principal value of $[1+(1-q)z]^{1/(1-q)}$ ($e_1^z=e^z$). Furthermore, the $q$-Fourier transform can also be defined for $q<1$ (see \cite{NelsonUmarov2010}).

Another important connection concerns the Dirac delta. The $\delta_q$ distribution is defined as \cite{JaureguiTsallis2010}
\begin{equation}
\delta_q(x)=\frac{2\pi}{2-q}\int_{-\infty}^{\infty}e_q^{i \xi x}\,d \xi \qquad (1\le q<2)\,,
\end{equation}
For $q=1$, this distribution is the usual plane-wave representation of Dirac delta. Also, the above integral corresponds, for an arbitrary $q$, to the $q$-FT of the constant function $f(x)=1$.
Wide families of functions exist \cite{JaureguiTsallis2010,ChevreuilPlastinoVignat2010,Mamode,PlastinoRocca} such that the $\delta_q$ distribution behaves like the Dirac delta for $q\not=1$ (see also section \ref{sec3}). 
\section{A method enabling the inversion of the $q$-Fourier transform}
Let $1\le q<2$ and $f(x)$ be a non-negative piecewise continuous function of the real variable $x$, whose support will be denoted by $\supp f$. Then, for each $y\in \supp f$, we define the function $f^{(y)}(x) = f(x+y)$ of the real variable $x$ (notice that $f^{(0)}(x)=f(x)$). The $q$-Fourier transform of $f^{(y)}(x)$ is given by
\begin {eqnarray}
F_q[f^{(y)}](\xi,y)=\int_{-\infty}^{\infty}f^{(y)}(x) e_q^{i \xi x[f^{(y)}(x)]^{q-1}}\,dx=\int_{-\infty}^{\infty}f(x+y)e_q^{i \xi x[f(x+y)]^{q-1}}\,d x\,.
\label {qFTxy}
\end {eqnarray}
Using the change of variables $z=x+y$, we have that
\begin{eqnarray*}
F_q[f^{(y)}](\xi,y)=\int_{-\infty}^{\infty}f(z)e_q^{i \xi(z-y)[f(z)]^{q-1}}\,dz=\int_{\supp f}f(z)e_q^{i \xi(z-y)[f(z)]^{q-1}}\,dz
\end{eqnarray*}
hence
\begin{eqnarray}
\label{int.permute}
\int_{-\infty}^{\infty}F_q[f^{(y)}](\xi,y)\,d \xi=\int_{-\infty}^{\infty}\int_{\supp f}f(z)e_q^{i \xi(z-y)[f(z)]^{q-1}}\,dz\,d\xi\,.
\end{eqnarray}
Assuming that the function $f(x)$ is such that it is allowed to commute the integral operators in the RHS of this equation, we have that
\begin{eqnarray*}
\int_{-\infty}^{\infty}F_q[f^{(y)}](\xi,y)\,d \xi=\int_{\supp f}f(z)\int_{-\infty}^{\infty}e_q^{i \xi(z-y)[f(z)]^{q-1}}\,d\xi\,dz=\frac{2\pi}{2-q}\int_{\supp f}f(z)\delta_q\left((z-y)[f(z)]^{q-1}\right)\,d z \, .
\end{eqnarray*}
Assuming also that the function $f(x)$ belongs to the class of functions for which the $\delta_q$ distribution behaves like the Dirac delta. Then we have
\begin{equation}
\label{step}
\int_{-\infty}^{\infty}F_q[f^{(y)}](\xi,y)\,d \xi=\frac{2\pi}{2-q}\int_{\supp f}f(z)\delta\left((z-y)[f(z)]^{q-1}\right)\,d z\,.
\end{equation}

Let us consider first the case in which $\supp f$ is a finite union of disjoint closed intervals, since in this case $\supp f$ has a boundary. Then $\supp f=\bigcup_{\ell=1}^m I_\ell$, where $I_\ell=[a_\ell,b_\ell]$ with $a_\ell<b_\ell$, and the $I_\ell$ are mutually disjoint. Then
\begin{eqnarray*}
\int_{-\infty}^{\infty}F_q[f^{(y)}](\xi,y)\,d \xi=\frac{2\pi}{2-q}\sum_{\ell=1}^m\int_{a_\ell}^{b_\ell}f(z)\delta\left((z-y)[f(z)]^{q-1}\right)\,dz=\frac{2\pi}{2-q}\sum_{\ell=1}^m\int_{a_\ell}^{b_\ell}f(z)\frac{\delta(z-y)}{[f(y)]^{q-1}}\,d z\,,
\end{eqnarray*}
where we used the property of the Dirac delta given in Appendix \ref{appendix}, since $f(x)>0$ for any $x\in\supp f$. As $y\in \supp f$ has been fixed, there exists only one $\ell_0$ such that $y\in[a_{\ell_0},b_{\ell_0}]$. Then
$$\int_{-\infty}^{\infty}F_q[f^{(y)}](\xi,y)\,d \xi=\int_{a_{\ell_0}}^{b_{\ell_0}}f(z)\frac{\delta(z-y)}{[f(y)]^{q-1}}\,d z=\frac{\gamma\pi}{2-q}[f(y)]^{2-q}\,.$$
Therefore,
\begin{equation}
f(y)=\left\{\frac{2-q}{\gamma\pi}\int_{-\infty}^{\infty}F_q[f^{(y)}](\xi,y)\,d \xi\right\}^{\frac{1}{2-q}}\,,
\label{mainresult}
\end{equation}
where $\gamma=2$ for any interior point $y\in \supp f$, and $\gamma=1$ for any boundary point $y\in \supp f$.

Let us now consider the case in which $\supp f=\mathbb{R}$. We straightforwardly obtain from Eq. (\ref{step}) that
\begin{eqnarray*}
\int_{-\infty}^{\infty}F_q[f^{(y)}](\xi,y)\,d \xi=\frac{2\pi}{2-q}\int_{-\infty}^\infty f(z)\frac{\delta(z-y)}{[f(y)]^{q-1}}\,dz=\frac{2\pi}{2-q}[f(y)]^{2-q}\,,
\end{eqnarray*}
which yields in the $(\gamma=2)$-particular case of Eq. (\ref{mainresult}).

If $\supp f$ is bounded either from above or from below, then Eq. (\ref{mainresult}) can be obtained from Eq. (\ref{step}) following a procedure similar to the one described on the above lines.
\paragraph{Particular case.} Let us remark that, in the $q \to 1$ limit, we have that
\begin{eqnarray*}
f(y)=\frac{1}{2\pi}\int_{-\infty}^{\infty}F[f^{(y)}](\xi,y)\,d \xi=\frac{1}{2\pi}\int_{-\infty}^{\infty}\int_{-\infty}^{\infty}f(x+y)e^{i \xi x}\,d x\,d\xi\,.
\end{eqnarray*}
Using the change of variables $z=x+y$, we have that
\begin{eqnarray}
f(y)=\frac{1}{2\pi}\int_{-\infty}^{\infty}\int_{-\infty}^{\infty}f(z)e^{i \xi(z-y)}\,dz\,d\xi =\frac{1}{2\pi}\int_{-\infty}^{\infty}e^{-i \xi y}\int_{-\infty}^\infty f(z)e^{i \xi z}\,dz\,d\xi=\frac{1}{2\pi}\int_{-\infty}^{\infty}e^{-i \xi y}F[f](\xi)\,d\xi\,,
\label{fourier}
\end{eqnarray}
which is the well-known expression of the inverse Fourier transform.

If we attempt to follow for $q \ne 1$ the above lines, we immediately find the difficulty that, in general, $e_q^{a+b} \ne e_q^a e_q^b$. Therefore, no generic explicit expression analogous to (\ref{fourier}) has been found for the inverse $q$-Fourier transform.

Eq. (\ref{mainresult}) says that we can obtain the function $f$ if we know the $q$-Fourier transforms of $f(x)$ \textit{and of all its translations} in the $x$-axis, which we illustrate next.

\begin{Ex}
Let $q<3$, $\beta>0$ and $x\in\mathbb{R}$. A $q$-Gaussian distribution is given (see for instance \cite{UmarovTsallisSteinberg2008,GellMannTsallis2004}) by 
\begin {equation}
G_{q,\beta}(x)=\frac{\sqrt{\beta}}{C_q}e_q^{-\beta x^2}\,,
\label {qgaussian}
\end {equation}
where $C_q$ is a normalization constant, in the sense that $C_q=\int_{-\infty}^{\infty}e_q^{-x^2}\,d x$. Let us now consider $G_{3/2,1}(x)$, then
\begin{equation}
G_{3/2,1}(y)=\left\{\frac{1}{4\pi}\int_{-\infty}^{\infty}F_q[G_{3/2,1}^{(y)}](\xi,y)\,d \xi\right\}^2
\label{fourierqg}
\end{equation}
for $1\le q<2$. The case $q=3/2$ was handled numerically (see Fig. \ref{fig.qGinverse}), since the $q$-Fourier transform of an arbitrary translation of $G_{3/2,1}(x)$ could not be obtained analitically, even Eq. (\ref{mainresult}) being an analytical result. We also verified that Eq. (\ref{fourierqg}) remains true for some other values of $q$.

\begin{figure}[t]
\centering
\includegraphics[width=0.45\textwidth,keepaspectratio]{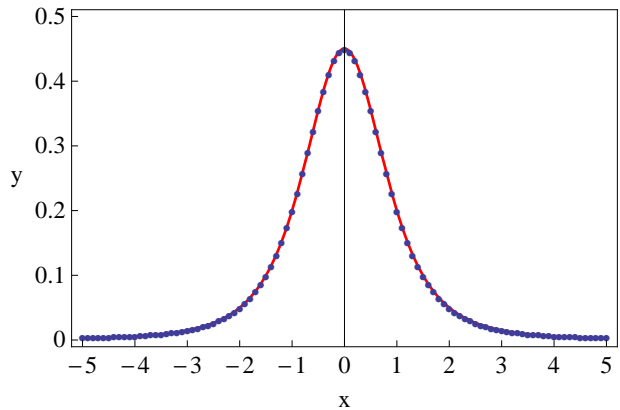}
\caption{Representation of $G_{3/2,1}(x)$. The continuous line corresponds to the analytical expression of the function; the dots were obtained by handling numerically Eqs. (\ref{qFT}) and (\ref{mainresult}).}
\label{fig.qGinverse}
\end{figure}
\end{Ex}

\begin{Ex}
Let $q\in(1,2)$, $A\geq 0$, $x\in\mathbb{R}$ and  
\begin {equation}
f_{A,q}(x)=\frac{C_q^{-1}\left[|x|^\frac{q-2}{q-1}-A\right]^\frac{1}{q-2}}{|x|^\frac{1}{q-1}\left[1+(q-1)\left(|x|^\frac{q-2}{q-1}-A\right)^\frac{2(q-1)}{q-2}\right]^\frac{1}{q-1}}
\label {hilhorst}
\end {equation}
when $A<|x|^{(q-2)/(q-1)}$, and zero otherwise (see \cite{Hilhorst2010}). This function is constructed such that $f_{0,q}(x)=G_{q,1}(x)$ and its $q$-FT does not depend on $A$. Thus, given only $F_q[f_{A,q}](\xi)=F_q[G_{q,1}](\xi)$ we cannot determine the original function $f_{A,q}$. Nevertheless, Eq. (\ref{mainresult}) states that, if we know the $q$-Fourier transform of $f_{A,q}(x)$ and of all its translations in the $x$-axis, then we can determine $f_{A,q}(x)$ (see Fig. \ref{fig.Hinverse}).
\begin{figure}[t]
\centering
\includegraphics[width=0.45\textwidth,keepaspectratio]{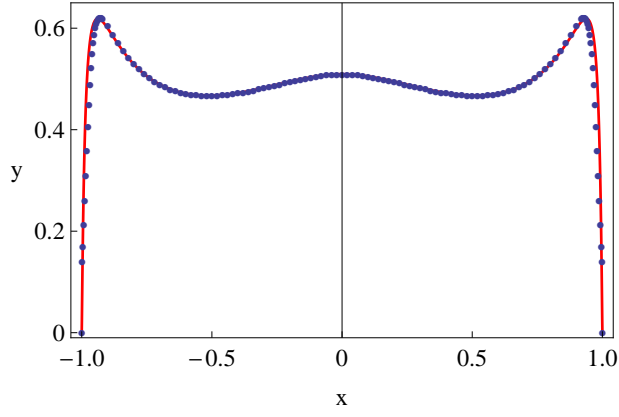}
\caption{Representation of $f_{1,5/4}(x)$. The continuous line corresponds to the analytical expression of the function; the dots were obtained by handling numerically Eqs. (\ref{qFT}) and (\ref{mainresult}). For all values of $x \in (-1,1)$ we have used  $\gamma=2$ in Eq. (\ref{mainresult}), whereas for $x=\pm 1$ we have used $\gamma=1$.}
\label{fig.Hinverse}
\end{figure}
\end{Ex}

\section{$\delta_q$ distribution as the Dirac delta}
\label{sec3}
Let us now discuss a family of functions for which the $\delta_q$ distribution behaves like the Dirac delta.

\begin{Thm}\label{L1func}
Let $q\in[1,2)$, and $f(x)$ be a function of the real variable $x$ which satisfies the following conditions:
\begin{enumerate}
\item [\textup{(i)}] $\int_{-\infty}^{\infty}|f(x)|\,d x<\infty$.
\item [\textup{(ii)}] $\int_{-\infty}^{\infty}|F[f](k)|\,d k<\infty$.
\end{enumerate}
Then, the $\delta_q$ distribution acts as the Dirac delta, i.e., 
$$\int_{-\infty}^{\infty}f(x)\delta_q(x)\,d x=f(0)\,.$$
\end{Thm}
\begin{proof}
If $q=1$, we have nothing to add to what is already available in the literature. If $q\in(1,2)$, we can use the Gamma representation of the $q$-exponential, namely \cite{ChevreuilPlastinoVignat2010},
\begin {equation}
e_q^{i \xi x}=\frac{1}{\Gamma\left(\frac{1}{q-1}\right)}\int_0^{\infty}e^{-t}t^\frac{2-q}{q-1}e^{i (q-1)\xi xt}\,d t\,,
\label {gamma}
\end {equation}
hence
\begin{eqnarray*}
\int_{-\infty}^{\infty}f(x)e_q^{i \xi x}\,d x= \frac{1}{\Gamma\left(\frac{1}{q-1}\right)}\int_{-\infty}^{\infty}\int_0^{\infty}f(x)e^{-t}t^{\frac{2-q}{q-1}}e^{i (q-1)\xi xt}\,dt\,dx\,.
\end{eqnarray*}
Notice that
\begin{eqnarray*}
\frac{1}{\Gamma\left(\frac{1}{q-1}\right)}\int_{-\infty}^{\infty}\int_0^{\infty}\left|f(x)e^{-t}t^\frac{2-q}{q-1}e^{i (q-1)\xi xt}\right|\,dt\,dx=\frac{1}{\Gamma\left(\frac{1}{q-1}\right)}\int_{-\infty}^{\infty}|f(x)|\,d x\int_0^{\infty}e^{-t}t^{\frac{2-q}{q-1}}\,dt\le \int_{-\infty}^{\infty}|f(x)|\,d x<\infty\,,
\end{eqnarray*}
where we have used condition (i). Then, from Fubini's theorem, we can permute the integral operators, hence
\begin{eqnarray*}
\int_{-\infty}^{\infty}f(x)e_q^{i \xi x}\,dx=\frac{1}{\Gamma\left(\frac{1}{q-1}\right)}\int_0^{\infty}e^{-t}t^{\frac{2-q}{q-1}}\int_{-\infty}^{\infty}f(x)e^{i (q-1)\xi xt}\,d x\,d t=\frac{1}{\Gamma\left(\frac{1}{q-1}\right)}\int_0^{\infty}e^{-t}t^{\frac{2-q}{q-1}}F[f]((q-1)\xi t)\,d t\,.
\end{eqnarray*}
We must now prove that
\begin {equation}
\int_{-\infty}^{\infty}\int_{-\infty}^{\infty}f(x)e_q^{i \xi x}\,d x\,d \xi=\frac{2\pi}{2-q}f(0)\,.
\label {proof}
\end {equation}
First, we should note that using the change of variables $\zeta=(q-1)\xi t$, we have that
\begin{eqnarray*}
\int_0^{\infty}\int_{-\infty}^{\infty}e^{-t}t^{\frac{2-q}{q-1}}|F[f]((q-1)\xi t)|\,d \xi\,dt&=&\int_0^{\infty}e^{-t}t^{\frac{2-q}{q-1}}\int_{-\infty}^{\infty}|F[f](\zeta)|\,\frac{d \zeta}{(q-1)t}\,d t\\
&=&\frac{1}{q-1}\int_0^{\infty}e^{-t}t^{\frac{3-2q}{q-1}}\,d t\int_{-\infty}^{\infty}|F[f](\zeta)|\,d \zeta<\infty\,,
\end{eqnarray*}
where we used condition (ii). Then, using again Fubini's theorem as well as $\zeta=(q-1)\xi t$, we have that
\begin{eqnarray*}
\int_{-\infty}^{\infty}\int_{-\infty}^{\infty}f(x)e_q^{i \xi x}\,d x\,d \xi&=&\frac{1}{(q-1)\Gamma\left(\frac{1}{q-1}\right)}\int_0^{\infty}e^{-t}t^{\frac{3-2q}{q-1}}\,d t\int_{-\infty}^{\infty}F[f](\zeta)\,d \zeta=\frac{2\pi f(0)}{(q-1)\Gamma\left(\frac{1}{q-1}\right)}\int_0^{\infty}e^{-t}t^{\frac{3-2q}{q-1}}\,d t\\
&=&\frac{2\pi f(0)\Gamma\left(\frac{2-q}{q-1}\right)}{(q-1)\Gamma\left(\frac{1}{q-1}\right)}=\frac{2\pi f(0)}{2-q}\,,
\end{eqnarray*}
which completes the proof.
\end{proof}

As a corollary, it follows from theorem \ref{L1func} that the $\delta_q$ distribution behaves like the Dirac delta for the family of functions $D_{a,b,\lambda}(x)=(ax^2+b)^\lambda$ of the real variable $x$ if $a,b>0$ and $\lambda<-1/2$. Indeed, we can straightforwardly verify that
$$\int_{-\infty}^{\infty}|D_{a,b,\lambda}(x)|\,dx\propto\int_{-\pi/2}^{\pi/2}(\cos\phi)^{-2(\lambda+1)}\,d\phi<\infty\,.$$
On the other hand (see \cite{Gradshteyn} (3.771.2))
$$F[D_{a,b,\lambda}](k)=2^{1+\lambda}\sqrt{\frac{2\pi}{a}}\left(\frac{\sqrt{ab}}{|k|}\right)^{\frac{1}{2}+\lambda}\frac{K_{\lambda+1/2}\left(\sqrt{\frac{b}{a}}|k|\right)}{\Gamma(-\lambda)}\,,$$
where $K_n(x)$ is the modified Bessel function of the second kind. We can verify, using the asymptotic expansion of $K_n(x)$ (see \cite{Gradshteyn} (8.451.6)), that $F[D_{a,b,\lambda}](k)$ is a positive integrable function. As a particular case, it follows that the $\delta_q$ distribution with $1\le q<2$ behaves like the Dirac delta for $q'$-Gaussians with $1\le q'<3$.

We analyzed numerically the behaviour of the $\delta_q$ distribution for the functions $D_{a,b,\lambda}(x)$ with $\lambda>0$. Let us remark that these functions behave smoothly at $x=0$, and diverge like a power law for $|x|\to \infty$. It results that there exists $q_{\rm max}(\lambda)\in[1,2)$ such that the $\delta_q$ distribution behaves like the Dirac delta for any $1\le q\le q_{\rm max}(\lambda)$ (see Fig. \ref{fig.divfunc}).

\begin{figure}[htp]
\centering
\includegraphics[width=0.45\textwidth,keepaspectratio]{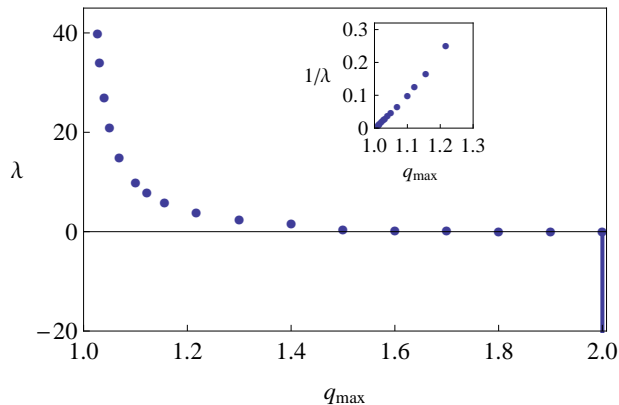}
\caption{For $1\le q\le q_{\rm max}(\lambda)$, the $\delta_q$ distribution behaves like the Dirac delta for functions of the type $D_{a,b,\lambda}(x)$, even when $\lambda>0$. The dots have been calculated numerically. The inset suggests that $q_{\rm max}\sim 1+1/\lambda$ ($\lambda\to \infty$).}
\label{fig.divfunc}
\end{figure}
\section{Conclusions}
Summarizing, we have shown that it is possible to determine the original distribution from the knowledge of the associated $q$-FT's for {\it all possible translations} $y$ of its argument, in contrast with the knowledge of {\it only} the $y=0$ case, as shown in \cite{Hilhorst2010}. The present results clarify some questions that emerge in connection with a large number of experimental, observational and computational results in physical systems which do suggest the emergence of $q$-Gaussian and $q$-exponential distributions (\cite{DouglasBergaminiRenzoni2006,LiuGoree2008,DeVoe2009,PickupCywinskiPappasFaragoFouquet2009,BurlagaNess2009,CarusoPluchinoLatoraVinciguerraRapisarda2007,CMS1,PHENIX,LyraTsallis1998,Borland2002} and others).

The exact mathematical conditions under which the exchange performed in Eq. (\ref{int.permute}) is legitimate remains as an interesting open point. Our examples reinforce however that it should be admissible for most physical cases.
\section*{Acknowledgements}
We acknowledge very fruitful discussions with E.M.F. Curado, H.J. Hilhorst, F.D. Nobre and S. Umarov. Partial financial support by Faperj and CNPq (Brazilian agencies) is acknowledged as well.
\appendix
\section{A property of the Dirac delta}
\label{appendix}
If $P(x)$ is a differentiable function of the real variable $x$, and $P(x)$ has a finite number of roots, namely $x_1,x_2,\ldots,x_n$, such that $P'(x_j)\ne 0$. Then
\begin{eqnarray*}
\delta(P(x))=\sum_{j=1}^n \frac{\delta(x-x_j)}{|P'(x_j)|}
\end{eqnarray*}
in the space of continuous functions of the real variable $x$.

\end{document}